\documentclass[11pt,a4paper,reqno]{amsart}
\usepackage{amsmath}
\usepackage[english]{babel}
\usepackage{latexsym}
\usepackage{amssymb}
\usepackage{amscd}
\usepackage{amsgen,amstext,amsbsy,amsopn}
\usepackage{math rsfs}
\usepackage{bm,bbm}
\usepackage{amsthm,epsfig,graphicx,graphics}
\usepackage[latin1]{inputenc}
\usepackage{xspace}
\usepackage{amsxtra}
\usepackage{color}
\usepackage{dsfont}
\usepackage{enumerate}

\usepackage{hyperref}

\usepackage{geometry}
\numberwithin{equation}{section}
\newcommand{\bdm}{\begin{displaymath}}
\newcommand{\edm}{\end{displaymath}}
\newcommand{\bay}{\begin{array}{c}}
\newcommand{\eay}{\end{array}}
\newcommand{\ben}{\begin{enumerate}}
\newcommand{\een}{\end{enumerate}}
\newcommand{\beq}{\begin{equation}}
\newcommand{\eeq}{\end{equation}}
\newcommand{\beqn}{\begin{eqnarray}}
\newcommand{\eeqn}{\end{eqnarray}}
\newcommand{\bml}[1]{\begin{multline} #1 \end{multline}}
\newcommand{\bmln}[1]{\begin{multline*} #1 \end{multline*}}

\newcommand{\lf}{\left}
\newcommand{\ri}{\right}

\newcommand{\xv}{\mathbf{x}}

\newcommand{\rv}{\mathbf{r}}

\newcommand{\nv}{\mathbf{n}}
\newcommand{\fv}{\mathbf{F}}

\newcommand{\diff}{\mathrm{d}}
\newcommand{\eps}{\varepsilon}

\newcommand{\dist}{\mathrm{dist}}

\newcommand{\as}{\alpha_{\star}}

\newcommand{\gav}{\bm{\gamma}}
\newcommand{\nuv}{\bm{\nu}}

\newcommand{\phiv}{\bm{\Phi}}

\newcommand{\glf}{\mathcal{G}^{\mathrm{GL}}}
\newcommand{\gle}{E^{\mathrm{GL}}}
\newcommand{\glm}{\psi^{\mathrm{GL}}}

\newcommand{\gldom}{\mathscr{D}^{\mathrm{GL}}}
\newcommand{\aav}{\mathbf{A}}
\newcommand{\aavm}{\mathbf{A}^{\mathrm{GL}}}

\newcommand{\hex}{h_{\mathrm{ex}}}
\newcommand{\theo}{\Theta_0}

\newcommand{\glfk}{\mathcal{G}_{\kappa}^{\mathrm{GL}}}
\newcommand{\glek}{E_{\kappa}^{\mathrm{GL}}}
\newcommand{\glmk}{\psi^{\mathrm{GL}}_{\kappa}}
\newcommand{\aavmk}{\mathbf{A}_{\kappa}^{\mathrm{GL}}}

\newcommand{\glfe}{\mathcal{G}_{\eps}^{\mathrm{GL}}}
\newcommand{\glfet}{\mathcal{G}_{\eps, \ann}^{\mathrm{GL}}}

\newcommand{\glee}{E_{\eps}^{\mathrm{GL}}}
\newcommand{\glff}{\mathcal{F}^{\mathrm{GL}}_{\eps}}

\newcommand{\eones}{E^{\mathrm{1D}}_{\star}}
\newcommand{\fs}{f_{\star}}
\newcommand{\onedom}{\mathscr{D}^{\mathrm{1D}}}

\newcommand{\curl}{\mathrm{curl}}

\newcommand{\ann}{\mathcal{A}_{\eps}}
\newcommand{\annt}{\tilde{\mathcal{A}}_{\eps}}
\newcommand{\acut}{\mathcal{A}_{\mathrm{cut}}}

\newcommand{\ocut}{\Omega_{\mathrm{cut}}}

\newcommand{\trial}{\psi_{\mathrm{trial}}}

\newcommand{\cellj}{\mathcal{C}_j}

\newcommand{\celldj}{\mathcal{D}_j}

\newcommand{\disp}{\displaystyle}
\newcommand{\tx}{\textstyle}

\newcommand{\Z}{\mathbb{Z}}
\newcommand{\R}{\mathbb{R}}

\newcommand{\F}{\mathcal{F}}
\newcommand{\E}{\mathcal{E}}

\newcommand{\OO}{\mathcal{O}}

\newcommand{\al}{\alpha}

\newcommand{\Om}{\Omega}

\newcommand{\one}{\mathds{1}}

\newcommand{\Hc}{H_{\mathrm{c}1}}
\newcommand{\Hcc}{H_{\mathrm{c}2}}
\newcommand{\Hccc}{H_{\mathrm{c}3}}
\newcommand{\Hstar}{H_{\mathrm{cc}}}

\newtheorem{teo}{Theorem}[section]
\newtheorem{lem}{Lemma}[section]
\newtheorem{pro}{Proposition}[section]

\newtheorem{asum}{Assumption}[section]

\theoremstyle{remark}
\newtheorem{remark}{Remark}[section]

\newcommand{\fone}{\E^{\mathrm{1D}}}

\newcommand{\eone}{E^{\mathrm{1D}}}

\newcommand{\eoneo}{E ^{\rm 1D}_{0}}

\newcommand{\jv}{\textbf{j}}

\begin{document}


\title{Surface Superconductivity in Presence of Corners}

\author[M. Correggi]{Michele CORREGGI}
\address{Dipartimento di Matematica e Fisica, Universit\`{a} degli Studi Roma Tre, L.go San Leonardo Murialdo, 1, 00146, Rome, Italy.}
\email{michele.correggi@gmail.com}

\author[E.L. Giacomelli]{Emanuela L. Giacomelli}
\address{Dipartimento di Matematica ``G. Castelnuovo'', ``Sapienza'' Universit\`{a} di Roma, P.le Aldo Moro, 5, 00185, Rome, Italy}
\email{giacomelli@mat.uniroma1.it}

\date{\today}

%

\begin{abstract}
We consider an extreme type-II superconducting wire with non-smooth cross section, i.e., with one or more corners at the boundary, in the framework of the Ginzburg-Landau theory. We prove the existence of an interval of values of the applied field, where superconductivity is spread uniformly along the boundary of the sample. More precisely the energy is not affected to leading order by the presence of  corners and the modulus of the Ginzburg-Landau minimizer is approximately constant along the transversal direction. The critical fields delimiting this surface superconductivity regime coincide with the ones in absence of boundary singularities.
\end{abstract}

\maketitle

\tableofcontents

\section{Introduction and Main Result}\label{sec:intro}


The success of Ginzburg-Landau (GL) theory in predicting and modelling the response of a superconductor to an external applied field can not be overlooked. Since the `50s, when is was proposed by \textsc{V.L. Ginzburg} and \textsc{L.D. Landau} \cite{GL} as a phenomenological theory to describe superconductivity near the critical temperature marking the transition to the superconducting state, several investigations have been performed within the GL theory from both the physical and mathematical view points (see, e.g., the monographs \cite{FH1,SS} and references therein).

Here we are mainly concerned with the phenomenon of {\it surface superconductivity}, which was predicted and observed in experiments in the `60s \cite{SJdG} (see also \cite{NSG} for recent experimental data). When the applied field becomes larger than some critical value ({\it second critical field}), the external magnetic field penetrates the sample and destroys superconductivity everywhere but in a narrow region close to the sample boundary.


The GL free energy of a type-II superconductor confined to an infinite cylinder of cross section $\Omega\subset\mathbb{R}^2$ is given by
\begin{equation}
	\label{eq: glfk}
	\glfk[\psi, \textbf{A}]= \displaystyle\int_\Omega \diff\textbf{r} \left\{ |(\nabla +i\hex\textbf{A})\psi|^2-\kappa^2|\psi|^2+ \tx\frac{1}{2}\kappa^2|\psi |^4 \ri\} + \hex^2\displaystyle\int_{\R^2}\diff\textbf{r}\; |\mbox{curl} \textbf{A} - 1|^2,
\end{equation}
where $\psi: \Omega \rightarrow \mathbb{C}$ is the order parameter, $\hex \textbf{A}:\mathbb{R}^2\rightarrow\mathbb{R}^2$ the vector potential, generating the induced magnetic field $ \hex \curl \aav $. The applied magnetic field is thus of uniform intensity $ \hex $ along the superconducting wire. The parameter $\kappa>0$ is a physical quantity (inverse of the penetration depth), which is typical of the material and an extreme type-II superconductor is identified by the condition $ \kappa \gg 1 $. We also recall the physical meaning of the order parameter $\psi$, i.e., $|\psi |^2$ is a measure of the relative density of superconducting Cooper pairs: $|\psi|$ varies between 0 and 1 and $|\psi |=0$ in a certain region means that there are no Cooper pairs there and thus a loss of superconductivity, whereas if $|\psi |=1$ somewhere then all the electrons are superconducting in the region. The cases $|\psi | \equiv 0$ and $|\psi| \equiv 1$ everywhere in $\Omega$ correspond to the \emph{normal} and to the \emph{perfectly superconducting states} respectively.

In the present paper we put the focus on {\it non-smooth} domains $ \Omega $, i.e., more precisely we assume that $ \Omega $ is piecewise smooth but contains finitely many {\it corners}. This is expected to generate a richer physics for large applied magnetic fields (see below), since the corners can act as attractors for Cooper pairs and, before eventually disappearing, superconductivity might survive close to them.

The equilibrium state of the superconductor given by the configuration $ (\glmk, \aavmk) \in \gldom $ is a minimizing pair for the functional $ \glfk $. The minimal energy will be denoted by $ \glek $ and the minimization domain can be taken to be
\beq
	\gldom = \lf\{ \lf( \psi, \aav\ri) \in H^1(\Omega) \times H^1_{\mathrm{loc}}(\R^2; \R^2) \: \big| \: \curl \aav - 1 \in L^2(\R^2) \ri\}.
\eeq
Any critical point $ (\psi, \aav) $ of the GL functional solves, at least weakly, the GL equations
\beq
	\label{eq: GL eqs}
	\begin{cases}
		- \lf( \nabla + i \hex \aav \ri)^2 \psi = \kappa^2 \lf(1 - |\psi|^2 \ri) \psi,		& \mbox{in } \Omega,		\\
		- \hex \nabla^{\perp} \curl \aav = \jv_{\aav}[\psi] \one_{\Omega},							& \mbox{in } \R^2,	\\
		\nv \cdot \lf( \nabla + i \hex \aav \ri) \psi = 0,									& \mbox{on } \partial \Omega,
	\end{cases}
\eeq
where the last boundary condition has to be interpreted in trace sense and
\beq
	\label{eq: current}
	\jv_{\aav}[\psi] : = \tx\frac{i}{2} \lf[ \psi \lf( \nabla - i \hex \aav \ri) \psi^* - \psi^* \lf(\nabla + i \hex \aav \ri) \psi \ri],
\eeq
is the superconducting current. Any minimizing configuration $ (\glmk, \aavmk) $ is a weak solution of the above system too.

The existence of such a minimizing pair and the fact that it solves the GL equations  \eqref{eq: GL eqs}  is a rather standard result \cite[Chapt. 15]{FH1} and we skip the discussion here. Note however that, because of the presence of a non-smooth boundary, it is more convenient to consider a GL functional where the last term is integrated all over the plane $ \R^2 $. In the smooth boundary case there is a one-to-one correspondence between minimizers on $ \R^2 $ and minimizing configurations of the energy restricted to $ \Omega $ (see, e.g., \cite[Proposition 3.4]{SS}) and therefore the two settings are perfectly equivalent. For non-smooth boundaries such a correspondence is difficult to state because of possible boundary singularities of the solution. It is easy to prove however \cite[Lemma 15.3.2]{FH1} that
\beq
	\curl \aav = 1,	\qquad		\mbox{in } \R^2 \setminus \overline{\Omega},
\eeq
for any $ \aav $ weak solution of the GL equations \eqref{eq: GL eqs}.

As the intensity of the applied field increases, one observes three subsequent transitions in an extreme type-II superconductor with smooth boundary and correspondingly three critical values of $ \hex $ can be identified, i.e., three {\it critical fields}:
\begin{itemize}
	\item the first critical field $ \Hc \propto \log \kappa $ is associated to the nucleation of vortices;
	\item at $ \Hcc = \kappa^2 $ the transition from bulk to boundary superconductivity occurs;
	\item above $ \Hccc = \theo^{-1} \kappa^2 $ the normal state is the unique minimizer of the GL functional and superconductivity is lost everywhere.
\end{itemize}

In presence of boundary singularities the above picture might require some modification and, although the first transition is untouched, being a bulk phenomenon, the two other transitions can be affected by the corners. More precisely it is known \cite{B-NF} that, under a suitable unproven conjecture about the ground state behavior of a magnetic Schr\"{o}dinger operator in an infinite sector, the third critical field $ \Hccc $ can be strictly larger than the one in the case of smooth domains, although of the same order $ \kappa^2 $. A more refined analysis reveals \cite[Theorem 1.4]{B-NF} that, at least if one of the corners has an acute angle, the expected value of $ \Hccc $ is
\beq
	\mu(\alpha)^{-1} \kappa^2,
\eeq
$ \mu(\alpha) $ standing for the ground state of $ - (\nabla + \frac{1}{2} i \xv^{\perp})^2 $, i.e., a Schr\"{o}dinger operator with uniform magnetic field of unit strength in an infinite sector of opening angle $ 0 < \alpha < \pi/2 $, and with $ \alpha $ equal to the smallest angle of the corners. The unproven assumption is that for any acute angle $ 0 < \alpha < \pi $, the eigenvalue $ \mu(\alpha) $ is strictly smaller than $ \theo $. For the discussion of this and many more questions concerning the Schr\"{o}dinger operator with magnetic field in a corner domain we refer to \cite{Bo,BD,BDP,BR}.

When (and if) $ \mu(\alpha) > \theo $ thus the third critical field changes value and in fact it is conjectured that another transition might take place below $ \Hccc $: as proven by usual Agmon estimate superconductivity can survive only at the boundary, if the field is above $ \Hcc $; however one might distinguish between a boundary state distributed along the boundary and one concentrated near the corner of smallest opening angle. Under the aforementioned conjecture this is indeed the structure of the GL ground state before the transition to the normal state takes place. As strongly suggested by the modified Agmon estimates proven in \cite[Theorem 1.6]{B-NF} one should introduce an additional critical field $ \Hstar $ so that
\beq
	\Hcc < \Hstar \leq \Hccc
\eeq
marking such a phase transition from {\it boundary} to {\it corner superconductivity}. The order of magnitude of $ \Hstar $ is clearly of order $ \kappa^2 $.

In the present paper we prove that, if $ \Omega $ contains finitely many corners, there is uniform surface superconductivity for applied fields satisfying asymptotically
\beq
	1 < \frac{\hex}{\kappa^2} < \theo^{-1},
\eeq
and the energy expansion is the same as in absence of corners, at least to leading order. As a consequence we infer the asymptotic estimate
\beq
	 \theo^{-1} \leq \frac{\Hstar}{\kappa^2} \leq \frac{\Hccc}{\kappa^2},
\eeq
which must hold true also in presence of corners.

Before going deeper in the discussion of our result, we first make a change of units, which is mostly convenient in the surface superconductivity regime, i.e., we assume that the applied field $ \hex $ is of order $ \kappa^2 $, i.e.,
\beq
	\hex = b \kappa^2,
\eeq
for some $ 0 < b = \OO(1) $, and set
\beq
	\varepsilon := b^{-\frac{1}{2}} \kappa^{-1}\ll 1 .
\eeq
We then study the asymptotic $\varepsilon \rightarrow 0$ of the minimization of the GL functional, which in the new units reads
\begin{equation}
	\label{eq: glf}
	\glfe [\psi, \textbf{A}]= \displaystyle\int_{\Omega} \diff\textbf{r}\; \bigg\{ \bigg| \left ( \nabla + i \frac{\textbf{A}}{\varepsilon ^ 2}\right)\psi \bigg|^2 -\frac{1}{2b\varepsilon ^2}(2|\psi|^2-|\psi|^4)	 \bigg\}+\frac{1}{\varepsilon ^4}\displaystyle\int_{\R^2} \diff\textbf{r}\; |\mbox{curl}\textbf{A} - 1|^2.
\end{equation}
We also set
\beq
	\glee := \displaystyle\min_{(\psi,\textbf{A})\in \gldom} \glfe[\psi,\textbf{A}],
\eeq
and denote by $ (\glm, \aavm) $ any minimizing pair.

In the new units the surface superconductivity regime coincides with the parameter region
\beq
	\label{eq: ss regime}
	1 < b < \theo^{-1},
\eeq
at least for smooth boundaries. The key features of the surface superconductivity phase are listed below:
\begin{itemize}
\item the GL order parameter is concentrated in a thin boundary layer of thickness $ \simeq \varepsilon $ and exponentially small in $ \eps $ in the bulk;
\item the induced magnetic field is very close to the applied one, i.e., $\mbox{curl}\textbf{A}\simeq 1$;
\item up to an appropriate choice of gauge and a mapping to boundary coordinates, the ground state of the theory can be approximated by an effective 1D energy functional in the direction normal to the boundary.
\end{itemize}

More precisely, in \cite{CR1}, it was proven\footnote{Note that, for the sake of clarity, we have changed the notation w.r.t. \cite{CR1} and denoted the 1D ground state energy $ \eones $, instead of $ \eoneo $ (the corresponding minimizer is then denoted by $ \fs $ instead of $ f_0 $). The reason is that the label $ 0 $ in the latter was actually referring to the dependence on the curvature, that we do not take into account in the present investigation.} that, if $\Omega $ is a smooth and simply connected domain, as $ \eps \to 0 $,
\beq
	\glee =\frac{|\partial \Omega | \eones}{\varepsilon}+\mathcal{O}(1),
\eeq
where $ \eones $ is the infimum of the functional
\begin{equation}
	\label{eq: fone}
	\fone_{\alpha}[f] := \displaystyle\int^{+\infty}_0 \mbox{dt} \left\{ |\partial_t f|^2 + (t+\alpha)^2 f^2 -\frac{1}{2b} (2f^2-f^4)\right\},
\end{equation}
both with respect to the real function $f$ and to the number $\alpha \in \R $, i.e.,
\beq
	\eones : = \inf_{\alpha \in \R} \inf_{f \in \onedom} \fone_{\alpha}[f],
\eeq
with $ \onedom = \lf\{ f \in H^1(\R; \R) \: | \: t f(t) \in L^2(\R) \ri\} $. The corresponding minimizing $ \alpha $ will be denoted by $ \as $, while $ \fs $ will stand for the function minimizing $ \eone_{\as} $.

The 1D functional \eqref{eq: fone} was known to provide a model problem for surface superconductivity since \cite{Pa}, but only in  \cite{CR1} (for disc-like samples) and \cite{CR2} (see also \cite{CR3} for further results in this direction) it was proven that $ \fs $ is a good approximation of $ |\glm| $. The main idea is indeed that any minimizing GL order parameter has the structure
\beq
	\label{eq: boundary behavior}
	\glm(\rv) \simeq \fs (\varepsilon t) \exp ( -i\as s ) \exp \lf\{ i \phi_\varepsilon  (s,t) \ri\}
\eeq
where $(s,t)$ are rescaled tubular coordinates (tangent and normal coordinates respectively, both rescaled by $ \eps $) defined in a neighborhood of $\partial \Omega$ (see, e.g., \cite[Appendix F]{FH1} and Section \ref{sec: coordinates}) and $\phi_\varepsilon$ is a gauge phase factor which plays an important role in the estimate of the magnetic field $ \aavm $.

Before stating our main result we specify the assumptions on the boundary domain: we consider a bounded domain $\Om\subset \R^2$ open and simply connected with Lipschitz boundary $\partial\Om$ such that the unit inward normal to the boundary, $ \nuv $, is well defined on $\partial \Om$ with the possible exception of a finite number of points -- the {\it corners} of $\Om$. We refer to the monographs \cite{Da,Gr} for a complete discussion of domains with non-smooth boundaries. More precisely we assume that the boundary $ \partial \Omega $ is a curvilinear polygon of class $C^{\infty}$ in the following sense (see also \cite[Definition 1.4.5.1]{Gr}):

	\begin{asum}[Piecewise smooth boundary]
		\label{asum: boundary 1}
		\mbox{}	\\
		Let $\Om$ a bounded open subset of $\R^2$, we assume that $ \partial \Omega $ is a smooth curvilinear polygon, i.e., for every $ \rv \in \Gamma$ there exists a neighborhood $ U $ of $ \rv $ and a map $ \phiv: U \to \R^2 $, such that
		\begin{enumerate}[(1)]
			\item $ \phiv $ is injective;
			\item $ \phiv $ together with $ \phiv^{-1} $ (defined from $ \phiv (U) $) are smooth;
			\item the region $\Om \cap U $ coincides with either $ \lf\{ \rv \in \Om \: | \: \lf( \phiv(\rv) \ri)_1 < 0 \ri\}$ or  $ \lf\{ \rv \in \Om \: | \: \lf( \phiv(\rv) \ri)_2 < 0 \ri\}$ or  $ \lf\{ \rv \in \Om \: | \: \lf( \phiv(\rv) \ri)_1 < 0, \lf( \phiv(\rv) \ri)_2 < 0 \ri\} $, where $ \lf( \phiv \ri)_j$ stands for the $ j-$th component of $ \phiv $.
		\end{enumerate}
	\end{asum}

	\begin{asum}[Boundary with corners]
		\label{asum: boundary 2}
		\mbox{}	\\
		We assume that the set  $\Sigma$ of corners of $ \partial \Omega $, i.e., the points where the normal $ \nuv $ does not exist, is non empty but finite and we denote by $\al_j$ the angle of the $j-$th corner (measured towards the interior).
	\end{asum}

	We are now able to formulate the main result proven in this article:

	\begin{teo}[GL asymptotics]
		\label{teo: GL asympt}
		\mbox{}	\\
		Let $\Om\subset\R^2$ be any bounded simply connected domain satisfying the above Assumptions \ref{asum: boundary 1} and \ref{asum: boundary 2}. Then for any fixed
		\beq
			\label{eq: b condition}
			 1< b < \theo^{-1},
		\eeq
		 as $\eps \to 0$, it holds
		\beq
			\label{eq: gle asympt}
			\gle = \frac{|\partial\Om| \eones}{\eps} + \OO(|\log\eps|^2),
		\eeq
		and
		\beq
			\label{eq: glm asympt}
			\lf\| \lf| \glm(\rv) \ri|^2 - \fs^2(\dist(\rv, \partial\Omega)/\eps) \ri\|_{L^2(\Omega)} = \OO(\eps |\log\eps|) \ll \lf\| \fs^2(\dist(\rv, \partial\Omega)/\eps) \ri\|_{L^2(\Omega)}.
		\eeq
	\end{teo}

	\begin{remark}[Order parameter asymptotics]
		\mbox{}	\\
		The convergence stated in \eqref{eq: glm asympt} implies that $ \fs $ provides a good approximation of $ |\glm| $ in the boundary layer, i.e., for $ \dist(\rv, \partial \Omega) \lesssim \eps |\log\eps| $. At larger distance from the boundary both functions are indeed exponentially small in $ \eps $ and their mass consequently very small. Note also that, if the condition \eqref{eq: b condition} is satisfied,
		\bdm
			\lf\| \fs^2(\dist(\rv, \partial\Omega)/\eps) \ri\|_{L^2(\Omega)} \geq c \sqrt{\eps},
		\edm
		for some $ c > 0 $, as it immediately follows by observing that $ \fs(t) $ is independent of $ \eps $ and non-identically zero.
	\end{remark}

	\begin{remark}[Limiting regimes]
		\mbox{}	\\
		We explicitly chose not to address the limiting cases $ b \to 1^+$ or $b \to \theo^{-1}$. In the former case an adaptation of the method might work (see also \cite[Remark 2.1]{CR1}), while in the latter the analysis is made much more complicate because of the interplay between corner and boundary confinements. In particular a much more detailed knowledge of the behavior of the linear problem, i.e., the ground state energy of the magnetic Schr\"{o}dinger operator in a sector of angle $ \alpha $, is needed, e.g., a proof of the conjecture discussed in \cite{B-NF}.
	\end{remark}

	The rest of the paper is devoted to the proof of the above result. Before proceeding however we briefly comment on possible future perspectives. Inspired by the analysis contained in \cite{CR2} one can indeed aim at capturing higher order corrections to the energy asymptotics \eqref{eq: gle asympt} and isolate the curvature corrections to the energy. At this level the presence of corners might give rise to an additional contribution to the energy of the same order of the curvature corrections. Compared to the smooth boundary case studied in \cite{CR2} however, a proof of an $ L^{\infty}$ estimate of $ |\glm| $ might be harder to obtain, due to the presence of boundary singularities. We plan to address this questions in a future work anyway.

\medskip

\noindent\textbf{Acknowledgments.} The authors acknowledge the support of MIUR through the FIR grant 2013 ``Condensed Matter in Mathematical Physics (Cond-Math)'' (code RBFR13WAET). M.C. also thanks \textsc{N. Rougerie} for enlightening discussions on the topic of the present paper.

\section{Proofs}

The strategy of the proof is very similar to the arguments contained in \cite{CR1}. We sketch here the main steps.

The preliminary step, i.e., the {\it restriction to the boundary layer} (Section \ref{sec: restriction}), is standard and described in details, e.g., in \cite[Section 14.4]{FH1}. The final outcome of this step is a functional restricted to a layer of width $ \OO(\eps|\log\eps|) $ along the boundary. The main ingredients are as usual Agmon estimates.

Another common step to both the upper and lower bound proofs, although applied to different magnetic potentials, is  the {\it replacement of the magnetic field} (Section \ref{sec: replacement}), as, e.g., in \cite[Appendix F]{FH1}. The presence of corners however calls for suitable modifications, since this step is usually done by exploiting tubular coordinates, which are not defined closed to the boundary singularities.  As we are going to see however, such a replacement is needed only in the smooth part of the boundary layer, where it can be done in a rather standard way by making a special choice of the gauge. In fact the only required modification is an adapted definition of the gauge phase close to the corners.

The energy {\it upper bound} (Section \ref{sec: upper bound}) is then trivially obtained by testing the energy on a trial configuration with magnetic field equal to the external one and with order parameter $ \psi $ reproducing \eqref{eq: boundary behavior}. The {\it lower bound} proof (completed in Section \ref{sec: lower bound}) is more involved and requires few more steps.

 In order to extract the 1D effective energy from the GL asymptotics, {\it boundary coordinates} (Section \ref{sec: coordinates}) are needed, since, according to the expected behavior \eqref{eq: boundary behavior}, the 1D energy is associated to the variation of $ |\glm| $ along the normal to the boundary, while $ |\glm| $ is approximately constant in the transversal direction. At this stage the non-smoothness of the boundary really affects the proof, because the use of tubular coordinates is clearly prevented near the corners. By a simple a priori estimate however we  show that one can {\it drop the energy around corners}.

We are thus left with the energy contributions of the smooth pieces of the boundary layer. There we can pass to boundary coordinates and use the same trick, i.e., a suitable {\it integration by parts}, involved in the proofs of our earlier results \cite{CR1,CR2} and inspired by other works on rotating condensates (see, e.g., \cite{CRv, CRY, CPRY1, CPRY2, CPRY3}). Since the region where we perform the integration by parts is not connected however, a naive application of the trick would generate unwanted boundary terms and therefore we will slightly modify the order parameter by introducing a partition of unity around the corners. 

Finally the key estimate to complete the lower bound proof is the {\it positivity of the cost function}, i.e., a pointwise estimate of an explicit 1D function depending on the effective 1D profile. This step is precisely the same as in \cite{CR1} and is the only point in the proof where the condition $ 1 < b < \theo^{-1} $ comes explicitly into play. However the assumption $ b > 1 $ is required to apply Agmon estimates in the preliminaries, while the condition $ b < \theo^{-1} $ is needed in order to ensure that the 1D minimizing profile is non-trivial.

We start the proof discussion by first recalling some properties of the 1D effective functional \eqref{eq: fone}, which is going to provide the effective energy in the surface superconductivity regime.

\subsection{Analysis of the effective $1D$ model}

Given the functional \eqref{eq: fone}, we denote by $f_{\alpha}$ any minimizer for $  \alpha \in \R $ and by $ \eone_{\alpha} $ the corresponding ground state energy,
$$
	\eone_{\alpha} := \inf_{f  \in H^1(\R^+)} \fone_{\al}[f],
$$
with the convention that $ \eones = \eone_{\as} = \inf_{\alpha \in \R} \eone_{\al} $ and $ \fs = f_{\as} $.

We recall that the minimizer $ \fs $ is non-trivial if and only if $ b < \theo^{-1} $ \cite[Proposition 3.2]{CR1} and, in this case, it is a unique positive function, which is also monotonically decreasing and solves the variational equation \cite[Proposition 3.1]{CR1}
\beq
	\label{eq: var eq fs}
	 -\fs'' +(t+ \as)^2 \fs = \tx\frac{1}{b}(1-\fs^2) \fs,
\eeq
with boundary condition $ \fs^{\prime}(0) = 0 $. The decay of $ \fs $ can be estimated \cite[Proposition 3.3]{CR1}: for any $ b < \theo^{-1} $, there exist two constants $ 0 < c,C < \infty $ such that
\beq
	\label{eq: fs decay}
	c\exp\left\{- \tx\frac{1}{2}(t+\sqrt{2})^2\ri\}\leq \fs(t)\leq C \exp\left\{- \tx\frac{1}{2}	(t+\alpha)^2\ri\}
\eeq
for any $t \in \R^+$. As a direct consequence 
\beq
	\fs(t) = \OO(\eps^{\infty}), \qquad	\mbox{for } t \geq c_0|\log\eps|,
\eeq
for any constant $ c_0 > 0 $.

The {\it potential function} associated to $ \fs $ is
\beq
	\label{eq: F}
	F(t) :=  2\displaystyle\int_0^t \diff \xi \; \fs^2(\xi) (\xi+\as).
\eeq
It can be shown that it is negative (in fact $ \as < 0  $) and vanishes both at $ t = 0  $ and $ t = \infty $. In this second case the vanishing of $ F $ is implied by the optimality of $ \as $. The {\it cost function} that will naturally appear in our investigation is
\beq
	\label{eq: K}
	K(t):= \fs^2(t) + F(t).
\eeq
In \cite[Proposition 3.4]{CR1} it is proven that
\beq
	\label{eq: K positive}
	K(t) \geq 0,		\qquad		\mbox{for any } t \in \R^+.
\eeq

\subsection{Restriction to the boundary layer}
\label{sec: restriction}

One of the key features of surface superconductivity is that the critical behavior survives only close to the boundary of the sample, due to the penetration of the external magnetic field. Mathematically this can be formulated by showing that the order parameter decays exponentially far from the boundary, which is the content of Agmon estimates. The presence of corners does not influence such a behavior.

We thus have that, if $ b > 1 $, any configuration $ (\psi,\aav) $ solving the GL equations satisfies the estimate
\beq
	\label{eq: agmon}
	\displaystyle\int_\Om \diff\textbf{r}\;\exp\lf\{\tx\frac{C_A \: \dist (\textbf{r},\partial\Om)}{\eps}\ri\} \lf\{|\psi|^2+\eps^2 \lf|\lf(\nabla+i \tx\frac{\textbf{A}}{\eps^2}\ri)\psi\ri|^2\ri\} \; \leq \displaystyle\int_{\{\dist(\textbf{r},\partial\Om)\leq \eps\}} \diff \textbf{r}\; |\psi|^2=\mathcal{O}(\eps),
\eeq
where $C_A$ is a fixed constant and in the last estimate we have used the bound $ \lf\| \psi \ri\|_{\infty} \leq 1  $, which remains true also in the case of corners \cite[Section 15.3.1]{FH1}. For the proof of normal Agmon estimates we refer to \cite[Theorem 4.4]{B-NF}. Note that because of the exploding exponential factor, \eqref{eq: agmon} implies that
\beq
	\int_{\dist(\rv, \partial \Omega)\geq c_0 \eps |\log\eps|} \diff \rv \: \lf\{ |\psi|^2+\eps^2 \lf|\lf(\nabla+i \tx\frac{\textbf{A}}{\eps^2}\ri)\psi\ri|^2 \ri\} = \OO(\eps^{c_0 C_A+1}),
\eeq
which can be made smaller than any power of $ \eps $ by taking the constant $ c_0 $ arbitrarily large. Thanks to this fact and the analogous estimate on the gradient of $ \psi $, one can easily drop from the energy the contribution of the region further from $ \partial \Omega $ than $ c_0 \eps |\log\eps| $ and, if we set
\beq
	\label{eq: bd layer}
	\ann : = \lf\{ \rv \in \Omega \: | \: \dist\lf(\rv, \partial \Omega\ri) \leq c_0 \eps |\log\eps| \ri\},
\eeq
it holds
\beq
	\label{eq: restriction}
	\gle = \glfet [\glm, \aavm] + \OO(\eps^{\infty}),
\eeq
where $ \glfet $ stands for
\beq
	\label{eq: glfet}
	\glfet[\psi,\aav] = \displaystyle\int_{\ann} \diff \rv \; \bigg\{ \bigg| \left ( \nabla + i \frac{\aav}{\varepsilon ^ 2}\right)\psi \bigg|^2 - \frac{1}{2b\varepsilon ^2}(2|\psi|^2-|\psi|^4)	 \bigg\}+\frac{1}{\varepsilon ^4}\displaystyle\int_{\Omega} \diff\textbf{r}\; |\mbox{curl}\textbf{A} - 1|^2.
\eeq

Before proceeding further we state some inequalities which will prove to be helpful in the rest of the paper (for a proof see, e.g., \cite[Section 15.3]{FH1} and reference therein): for any weak solution to the GL equations \eqref{eq: GL eqs}
\beqn
	\lf\|\psi\ri\|_{L^\infty (\Om)} & \leq & 1,	\label{eq: est infty}	\\
	\lf\| \lf(\nabla + i\tx\frac{\aav}{\eps^2} \ri) \psi \ri\|_{L^2(\Om)} & = & \OO(\eps^{-1/2}),	\label{eq: est gradient}	\\
	\lf\| \curl \aav - 1 \ri\|_{L^2(\R^2)} & = & \OO(\eps^{7/4}).	\label{eq: est curl a}
\eeqn

\subsection{Boundary coordinates}
\label{sec: coordinates}

The rest of the proof requires the definition of suitable boundary coordinates $ (\sigma,\tau) $, e.g., to isolate the singular regions around the corners. As in \cite[Appendix F]{FH1} we introduce a parametrization of the boundary $ \partial \Omega $ denoted by $ \gav(\sigma) $, $ \sigma \in [0,\partial\Omega) $, which is piecewise smooth. At any point along the boundary, with the exception of corners $ \Sigma $, the inward normal to the boundary $ \nuv(\sigma) $ is well defined and smooth.

\begin{figure}[ht!]
\begin{center}
\includegraphics[scale=0.7]{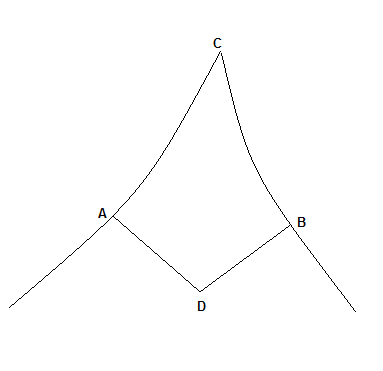}
\caption{Region where boundary coordinates might be ill defined.}
\label{fig: cell 1}
\end{center}
\end{figure}

The following map however
\beq
	\rv(\sigma,\tau) = \gav(\sigma) + \tau \nuv(\sigma),
\eeq
with $ \tau = \dist(\rv, \partial \Omega) $, defines a diffeomorphism only in a thin enough strip along $ \partial \Omega $ and far enough from $ \Sigma $, e.g., in
\bdm
	\lf\{ \rv \: | \: \dist\lf(\rv, \partial \Omega\ri) \leq \tau_0, \dist\lf(\rv, \Sigma\ri) \geq \tau_0^{\prime} \ri\},
\edm
with $ \tau_0 $ small enough and $ \tau_0^{\prime} $ suitably chosen (of the same order). In that region we can also define the curvature $ \varsigma(\sigma) $ as
\bdm
	\gav^{\prime\prime}(\sigma) = \varsigma(\sigma) \nuv(\sigma).
\edm

In order to handle the singularities at corners we need to define cells covering the region where tubular coordinates are ill defined. Given the normal width of the boundary layer equal to $ c_0 \eps |\log\eps| $, as required to apply Agmon estimates in the complement, it is obviously possible to find a corner cell of the form given in Fig. \ref{fig: cell 1}. Since the length of inner boundaries is $ c_0 |\log\eps |$, the distances between the points $ A $ and $ B $ and the corner $ C $ along the boundary can be chosen at most of order $ \OO(\eps|\log\eps|) $.

\begin{figure}[ht!]
\begin{center}
\includegraphics[scale=0.7]{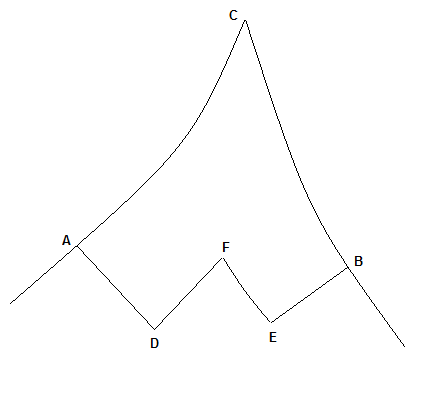}
\caption{Cell $ \cellj $.}
\label{fig: cell 2}
\end{center}
\end{figure} 

 For the sake of concreteness we fix those lengths (the curves $ AC $ and $ BC $ in Fig. \ref{fig: cell 2}) both equal to
\beq
	\sigma_j - \sigma_A  = \sigma_B - \sigma_j = c_1 \eps|\log\eps|,	
\eeq
where the constant $c_1 $ is chosen large enough in such a way that in $ (\cup_j \cellj)^c $ tubular coordinates $ (\sigma,\tau) $ are well defined. The inner boundary of the cell is given by the two lines ($ AD $ and $ BE $ in Fig. \ref{fig: cell 2})
\beq
 	 \lf\{ \rv(\sigma,\tau) \: | \: \sigma = \sigma_j \pm c_1\eps|\log\eps|, t \in [0, c_0 \eps|\log\eps|] \ri\},
\eeq
with possibly two segments ($ DF$ and $ EF $ in Fig. \ref{fig: cell 2}) belonging to the border of $ \ann $ of the form
\beq
	\lf\{ \rv(\sigma,\tau) \: | \: \lf| \sigma - \sigma_j \ri| \in \lf[\delta_{\pm}\eps|\log\eps|, c_1\eps|\log\eps| \ri],  t = c_0 \eps|\log\eps|] \ri\},
\eeq
for some $ 0 < \delta_{\pm} < c_1 $. The final shape of the cell is described in Fig. \ref{fig: cell 2} for an acute angle. The definition requires no adaptation however for obtuse angles.

The most important property of the corner cells is that they carry a little amount of energy, which allow us to discard them in the estimate of both the upper and lower bounds to the GL energy. This is directly implied by the smallness of those cells, whose area is $ \OO(\eps^2 |\log\eps|^2) $.  At an heuristic level indeed the energy density is of order $ \eps^{-2} $, at least close to $ \partial \Omega $, and therefore the energy contained in $ \cellj $ is expected to be $ \OO(|\log\eps|^2) $, i.e., of the same order to the error term in \eqref{eq: gle asympt}.

In the rest of the proof we will also use the rescaled boundary coordinates defined in terms of $ (\sigma, \tau) $ as
\beq
	s = \frac{\sigma}{\eps},	 \qquad	t = \frac{\tau}{\eps},
\eeq
so that, with some abuse of notation\footnote{For the sake of simplicity in the notation, we do not change symbol to denote a given set in the rescaled variables, e.g., $ \ann $ will stand for the boundary layer, either in the $ \rv $ or in the $ (s,t) $ variables.},
\beq
	\ann = \lf\{ (s,t) \: | \: s \in \lf[0, \tx\frac{|\partial \Omega|}{\eps}\ri], t \in [0, c_0|\log\eps|] \ri\}.
\eeq
The boundary layer without the corner cells is called $ \acut $, i.e.,
\beq
	\label{eq: acut}
	\acut : = \ann \setminus \bigcup_{j = 1}^N \cellj.
\eeq
We also denote by $ \nv(s) : = \nuv(\eps s) $ and $ k(s) : = \varsigma(\eps s) $ the normal to $ \partial \Omega  $ and the curvature in the new coordinates respectively.
	

\subsection{Replacement of the magnetic field}
\label{sec: replacement}

In the surface superconductivity regime the induced magnetic field is actually very close to the applied one. This is typical of this regime and is also the reason why the last term in the GL functional \eqref{eq: glf} is never taken into account, since its contribution is always subleading. Hence the induced field is almost uniform and its strength is approximately $ 1$. There are however many magnetic potentials $ \aav $ generating such a field and it is useful to exploit gauge invariance to select the most convenient one. Here we discuss how it can be done. 

We pick a magnetic potential $ \aav $ such that 
\beq
	\label{eq: coulomb gauge}
	\nabla \cdot \aav = 0,	\qquad		\mbox{in } \Omega,
\eeq
i.e., in the Coulomb gauge, and \eqref{eq: est curl a} holds true,
\beq
	\label{eq: a assumption}
	\lf\| \curl \aav - 1 \ri\|_{L^2(\R^2)} = \OO(\eps^2).
\eeq
This in particular applies to the minimizer $ \aavm $ but also the magnetic potential used in the upper bound (see Section \ref{sec: upper bound} below) satisfies the condition. 

In smooth domains \eqref{eq: coulomb gauge} is accompanied by the boundary condition $ \aav \cdot \nv = 0 $ on $ \partial \Omega $. This is clearly not possible in presence of corners, because of the jumps of $ \nv $: let us stress that at corners $ \nv $ is discontinuous but it remains uniformly bounded all over the boundary, so that, for instance,
\bdm
	\int_{0}^{\frac{|\partial \Omega|}{\eps}} \diff s \: \aav(\rv(\eps s,0)) \cdot \nv(s) = 0,
\edm
i.e., the function $ \aav \cdot \nv $ is integrable. Moreover Stokes formula and elliptic regularity (see, e.g., \cite{Gr}) implies that the boundary condition is in fact satisfied in trace sense, i.e., almost everywhere. In particular it must hold true far from corners, where the normal to the boundary is well defined, e.g.,
\beq
	\aav(\rv(\eps s,0)) \cdot \nv(s) = 0,	\qquad		\mbox{for any } |s-s_j| \geq c_1 |\log\eps|.
\eeq

Now the key idea for the replacement, described, e.g., in \cite[Appendix F]{FH1} (see also \cite[Section 4.1]{CR1} and \cite[Proof of Lemma 4]{CR2}), is that, since at the boundary the normal component of $ \aav $ vanishes, close enough to $ \partial \Omega $, it is possible to find a differentiable function $ \phi_{\aav}(s,t) $ such that $ \aav - \nabla \phi_{\aav}(s,t) = f(s,t) {\bf e}_s $, i.e., the field remains purely tangential close enough to $ \partial \Omega $. In addition since $ \curl \aav $ is approximately $ 1 $, the function $ f(s,t) $ is close to $ - t + o(1) $. Using a gauge transformation built on $ \phi_{\aav} $, one can then replace $ \aav $ with a magnetic potential which is of the form described above.

We do not repeat all the details of the procedure but we only comment on the modifications induced by the presence of corners. The first step, i.e., the gauge transformation is essentially the same as for smooth domains: thanks to the integrability of both $ \gav^{\prime}(s) $ and $ \aav(\rv(\eps s, \eps t)) $ along $ \partial \Omega $, which is guaranteed in the first case by the boundedness of $ \gav^{\prime}(s) $ and in the second one by Sobolev trace theorem, which gives $ \aav \in H^{3/2}(\partial \Omega; \R^2) $, we can set for any $ (s,t) \in \acut $
\beq
	\label{eq: phase}
	\phi_\aav (s,t) := - \frac{1}{\eps} \int_0^t \diff\eta \: \aav(\rv(\eps s, \eps \eta)) \cdot \nv(s)  + \frac{1}{\eps} \int _0^{s} \diff \xi \;  \aav(\rv(\eps\xi, 0)) \cdot \gav^{\prime}(\eps s) - \eps \delta_\eps s,
\eeq
with\footnote{We denote by $ \lfloor \: \cdot \: \rfloor $ the integer part. Note the missing factors $ 2\pi $ in the definitions \cite[Eq. (4.8)]{CR1} and \cite[Eq. (5.4)]{CR2}.}
\beq 
	\delta_\eps (s,t) = \frac{1}{\eps^2|\partial\Om|} \int_{\Om} \diff \rv \; \curl \aav - \frac{2\pi}{|\partial\Om |}  \lf \lfloor  \frac{1}{ 2 \pi \eps^2} \int_{\Om } \diff \rv \; \curl \aav \ri\rfloor.
\eeq
Note that the second term in the expression above is well defined even if $ s $ gets close to the corner points $ \Sigma $, because $ \gav^{\prime} $ is by assumption an integrable function. In $ \acut^c $ on the other hand the phase $ \phi_\aav $ can be continued in a rather arbitrary way with the only condition that $ \phi_{\aav} \in  H^1(\ann) $, e.g., we can set for $ (s,t) \in \acut^c $
\beq
	\label{eq: phase corners}
	\phi_\aav (s,t) := - \frac{1}{\eps} \int_0^t \diff\eta \: \chi(s) \aav(\rv(\eps s, \eps \eta)) \cdot \nv(s)  + \frac{1}{\eps} \int _0^{s} \diff \xi \;  \aav(\rv(\eps\xi, 0)) \cdot \gav^{\prime}(\eps s) - \eps \delta_\eps s,
\eeq
where the cut-off function $ \chi $ vanishes at distance of order $ 1 $ from the boundaries $ s_j \pm c_1 |\log\eps| $. It is easy to verify that for any $ (s,t) \in \ann $, one has 
\bdm
	\phi_{\aav}(s + n |\partial \Omega|, t) = \phi_{\aav}(s,t) + 2 \pi n,	\qquad \mbox{for any } n \in \Z.
\edm

 The change of coordinates $ \rv \to (\eps s, \eps t) $ in $ \acut  $ and the simultaneous gauge transformation then yields for any $ \Psi \in H^1(\ann) $
\bml{
	\glf_{\eps,\acut}[\Psi,\aav] = \lf(1 + \OO(\eps|\log\eps|) \ri) \int_{\acut} \diff s\diff t \; \lf\{ |\partial_t \psi |^2 + \lf| (\partial_s + i   a(s,t)) \psi \ri|^2 \ri.	\\
	\lf. - \tx\frac{1}{2b}[2|  \psi|^2-| \psi|^4] \ri\}	+ \frac{1}{\eps^4} \int_{\Omega} \diff \rv \: \lf| \curl \aav - 1 \ri|^2,
}
where the prefactor $ 1 + \OO(\eps|\log\eps|) $ is due to an estimate of the jacobian $ 1 - \eps k(s) t $ of the change of coordinates $ \rv \to (s,t) $ induced by the diffeomorphism $ \rv(\eps s, \eps t) $, where $ k(s) $ is the curvature of the boundary in the rescaled coordinates. Of course $ k(s) $ is not defined at corners but admits left and right values and 
\beq
	\lf\| k \ri\|_{L^{\infty}(\partial \Omega)} \leq C.
\eeq
Moreover $ \psi $  in the expression above is 
\beq
	\psi(s,t) : = \Psi(\rv(\eps s, \eps t)) e^{-i \phi_{\aav}(s,t)}.
\eeq
Finally
\beq
	\label{eq: magn pot a}
	a_{\aav}(s,t) = (1-\eps k(s)t) \frac{\gav^{\prime}(\eps s) \cdot \aav(\rv(\eps s,\eps t))}{\eps} - \partial_ s \phi_\aav.
\eeq
Note that we have left untouched the last term of the GL functional, because it will be treated in different ways in the upper and lower bound proofs. 

We also stress that so far the a priori assumption \eqref{eq: a assumption} has not been used. It is indeed the main ingredient of the estimate of $ a(s,t) $.

 	\begin{lem}
 		\label{lem: a est}
 		\mbox{}	\\
 		Let $ \aav  $ be such that \eqref{eq: coulomb gauge} and \eqref{eq: a assumption} are satisfied, then for any $ c_0 > 0 $,
 		\beq
 			\label{eq: tilde a est} 
 			\lf\| {a}_{\aav}(s,t)  + t \ri\|_{L^2(\acut)} = \OO(\eps|\log\eps|).
		\eeq
 	\end{lem}		
 			
 	\begin{proof}
 		Let $ (s,t) \in \acut $, we first observe that
		$$
			{a}_{\aav} (s,0)= \eps\delta_\eps = \OO(\eps),
		$$
		since $ |\delta_{\eps}| \leq 1 $.  The definition of $  a_{\aav} $ and the vanishing of the normal component also implies that 
		\bdm
			\partial_t{a}_{\aav} (s,t) = -(1-\eps k(s)t)(\curl \aav) (\rv(\eps s,\eps t))
		\edm
		and therefore
		\bml{
			{a}_{\aav} (s,t) =  \eps\delta_\eps - (1 + \OO(\eps|\log\eps|) \displaystyle\int_0^t \diff \eta \; \curl\aav(\rv(\eps s,\eps \eta)) \\
			= - t - \displaystyle\int_0^t \diff \eta \; \lf[ \curl\aav(\rv(\eps s,\eps \eta)) - 1 \ri] + \OO(\eps|\log\eps|).
		}
		On the other hand we can estimate
		\bdm
			\int_0^t \diff\eta\; \lf| \curl \aav(\rv(\eps s,\eps \eta)) - 1 \ri | \leq C |\log\eps|^{\frac{1}{2}} \bigg[ \int_{0}^{c_0|\log\eps|} \diff t \: \lf| \curl \aav  - 1 \ri|^2 \bigg]^{1/2}
		\edm
		which yields
		\bdm
			\lf\|  a_{\aav} + t \ri\|_{L^2(\acut)}^2 \leq C |\log\eps| \lf\| \curl \aav -1 \ri\|_{L^2(\acut)}^2 + \OO(\eps^2|\log\eps|^2),
		\edm
		and therefore the result.
	\end{proof}

\subsection{Upper bound}
\label{sec: upper bound}

We are now able to prove the upper bound to the GL energy.

	\begin{pro}[Energy upper bound]
		\label{pro: upper bound}
		\mbox{}	\\
		Let $1<b<\Theta_0^{-1}$ and $\eps$ be small enough. Then it holds
		\beq
			\label{eq: upper bound}
			\gle \leq \frac{|\partial \Om| \eones}{\eps} + \OO(|\log\eps|).
		\eeq
	\end{pro}

	\begin{proof}
		As usual we prove the result by evaluating the GL energy on a trial state having the expected physical features. This trial state has to be concentrated near the boundary of the sample and its modulus must be approximately constant in the transversal direction. However since we can not use boundary coordinates at corners, we impose that the function vanishes in a suitable neighborhood of $ \Sigma $. We will label corners in $ \Sigma $ with their coordinate along the boundary, i.e.,
		\beq
			\Sigma = \lf\{ (s,0) \: | \: s = s_j, 	j = 1, \ldots, N \ri\}.
		\eeq

		We thus introduce a cut-off function $\chi \in C^\infty_0$, such that $ 0 \leq \chi \leq 1$. Its role is to cut the region close to $ \Sigma $. For any $ (s,t) \in \ann $, we require
		\beq
			\label{eq: chi}
			\chi(s) := 
			\begin{cases}

				0,	&	\mbox{if } \lf| s - s_j \ri| \leq c_1|\log\eps|,		\mbox{ for some } j = 1, \ldots, N, \\

				1,	&	\mbox{if } \lf| s - s_j \ri| \geq 2c_1|\log\eps|, \mbox{ for all } j = 1, \ldots, N.
			\end{cases}
		\eeq
		The transition from $ 0 $ to $ 1 $ occurs in a one-dimensional region of length $ c_1 |\log\eps| $ and therefore we can always assume that
		\beq
			\label{eq: chi gradient}
			\lf| \chi^{\prime} \ri| = \OO(|\log\eps|^{-1}).
		\eeq
		We also define (note that we can not define $ \celldj $ using boundary coordinates in the interior because there they are ill defined)
		\beq
			\celldj : = \lf\{ (s,t) \in \ann \: | \: \lf| s - s_j \ri| \geq 2c_1|\log\eps| \ri\}^c,
		\eeq
		i.e., $ (\cup_j \celldj)^c $ is the region where $ \chi = 1 $.
		
		It remains to choose the magnetic potential to complete the test configuration: we thus denote by $ \fv $ any magnetic potential such that $ \nabla \cdot \fv = 0  $ and $ \curl \fv = 1 $ in $ \Omega $. Our trial state is then
		\beq
			\lf( \trial, \fv \ri),
		\eeq
		where
		\beq
			\label{eq: trial}
			\trial(s,t) = \chi(s) \fs(t) e^{-i  \as s} e^{i \phi_{\fv}(s,t)},
		\eeq
		with $ \phi_{\aav} $ the gauge phase \eqref{eq: phase}. Note that $ \as $ is not necessarily an integer and therefore $ \trial $ might be a multi-valued function, but this does not harm the result since we are here interested in proving an upper bound.
		
		The order parameter decays exponentially as $ t \to \infty $, thanks to the pointwise estimate \eqref{eq: fs decay}, and therefore 
		\bdm
			\trial(\rv) = \OO(\eps^{\infty}),		\qquad \mbox{for } \rv \in \ann^c.
		\edm
		Hence, as for \eqref{eq: restriction}, we have
		\beq
			\glfe[\trial,\fv] = \glf_{\eps, \ann}[\trial, \fv] + \OO(\eps^{\infty}).
		\eeq
		Also, since $ \curl\fv = 1 $, the quantity we have to estimate is actually 
		\bml{
			\label{eq: ub eq 1}
			\glf_{\eps, \ann}[\trial, \fv] = \displaystyle\int_{\ann} \diff \rv \; \lf\{ \lf| \left ( \nabla + i \tx\frac{\fv}{\varepsilon ^ 2}\ri) \trial \ri|^2 - \tx\frac{1}{2b\varepsilon ^2}(2|\trial|^2-|\trial|^4)	 \ri\}	\\
			= \lf(1 + \OO(\eps|\log\eps|) \ri) \int_{\ann} \diff s\diff t \; \lf\{ \chi^2 |\partial_t \fs |^2 + \fs^2 \lf| \partial_s \chi \ri|^2 + \lf| a_{\fv}(s,t) - \as \ri|^2 \fs^2 \chi^2 \ri.	\\
	\lf. - \tx\frac{1}{2b}[2 \chi^2 \fs^2 - \chi^4 \fs^4] \ri\}	
		}
		with $ a_{\fv} $ defined in \eqref{eq: magn pot a}. For an upper bound we can replace $ \chi $ with $ 1 $ wherever $ \chi \not = 0 $, i.e.,
		\beq
			\label{eq: ub eq 2}
			\glf_{\eps, \ann}[\trial, \fv] \leq \int_{\ann \setminus \cup_j \cellj} \diff s\diff t \; \lf\{ {\fs^{\prime}}^2 + (t + \as )^2 \fs^2 - \tx\frac{1}{2b}[2\fs^2-\fs^4] \ri\}	 + \OO(|\log\eps|),
		\eeq
		where the remainder is due to various factors, we are going to explain. We first notice that
		\bml{
			\label{eq: ub eq 3}
			\int_{\ann \setminus \cup_j \cellj} \diff s\diff t \; \lf\{ {\fs^{\prime}}^2 + (t + \as )^2 \fs^2 - \tx\frac{1}{2b}[2\fs^2-\fs^4] \ri\} \\
			= \int_{\frac{\partial \ocut}{\eps}} \diff s \int_0^{c_0|\log\eps|} \diff t \; \lf\{ {\fs^{\prime}}^2 + (t + \as )^2 \fs^2 - \tx\frac{1}{2b}[2\fs^2-\fs^4] \ri\}	\\
			= \frac{\lf|\partial \ocut \ri| \eones}{\eps} \leq \frac{|\partial \Omega| \eones}{\eps},
		}
		where $ \partial \ocut = \partial \Omega \cap \partial \acut $, so that the prefactor $ \OO(\eps |\log\eps|) $ in \eqref{eq: ub eq 1} generates an error of order $ |\log\eps| $.  Moreover
		\bmln{
			 -\frac{1}{2b} \int_{\ann} \diff s\diff t \; [2 \chi^2 \fs^2 - \chi^4 \fs^4] \leq -\frac{1}{2b} \int_{\ann \setminus \cup_j \cellj} \diff s\diff t \; [2 \fs^2 - \fs^4] + C \int_{\ann} \diff s\diff t \;\lf(1 - \chi^2  \ri) \fs^2 	\\
			 = -\frac{1}{2b} \int_{\ann \setminus \cup_j \cellj} \diff s\diff t \; [2 \fs^2 - \fs^4] + \OO(|\log\eps|),
		}	
		and using that the exponential decay of $ \fs $ given by \eqref{eq: fs decay},
		\bmln{
			\int_{\ann} \diff s\diff t \; \lf| a_{\fv}(s,t) - \as   \ri|^2 \fs^2 \chi^2 \leq \int_{\ann \setminus \cup_j \cellj} \diff s\diff t \;  (t + \as)^2 \fs^2  	\\
			+ C \lf\| a_{\fv}(s,t) + t \ri\|_{L^2(\acut)}	\lf\| (t + \as) \fs \ri\|_{L^2(\ann)} + C \lf\| a_{\fv}(s,t) + t \ri\|_{L^2(\acut)}^2	\\
			=  \int_{\ann \setminus \cup_j \cellj} \diff s\diff t \; (t + \as)^2 \fs^2 + \OO(\sqrt{\eps}|\log\eps|).
		}
		Finally the kinetic energy of the cut-off is bounded as
		\bdm
			\int_{\ann} \diff s\diff t \; \fs^2 \lf| \partial_s \chi \ri|^2 \leq C |\log\eps|^{-2} \lf| \cup_j \lf( \celldj \setminus \cellj \ri) \ri| = \OO(1),
		\edm
		thanks to the assumption \eqref{eq: chi gradient}.
		
		Combining \eqref{eq: ub eq 2} and \eqref{eq: ub eq 3} the energy upper bound is proven.
	\end{proof}

\subsection{Lower bound and completion of the proof}
\label{sec: lower bound}

The first step towards a proof of a suitable lower bound is the control of the energy contributions of corners. This is however rather easy to obtain since
\bml{
	\label{eq: pre lower bound}
	\glf_{\eps, \ann}[\glm,\aavm] \geq \glf_{\eps, \acut}[\glm,\aavm] + \OO(|\log\eps|^2) \\
	\geq \glff[\glm,\aavm]+ \OO(|\log\eps|^2),
}
where $ \acut $ is given in \eqref{eq: acut},
\beq
	\label{eq: glff}
	\glff[\psi,\aav] : = \int_{\acut} \diff \rv \; \lf\{ \lf| \left ( \nabla + i \tx\frac{\aav}{\varepsilon ^ 2}\ri) \psi \ri|^2 - \tx\frac{1}{2b\varepsilon ^2}(2|\psi|^2-|\psi|^4)	 \ri\}
\eeq
and the remainder is produced by the only non-positive term of the GL functional i.e.,
\bdm
	- \frac{1}{b \eps^2} \int_{\ann} \diff \rv \: |\glm|^2 \geq - \frac{1}{b \eps^2} \int_{\acut} \diff \rv \: |\glm|^2 - C |\log\eps|^2,
\edm
by \eqref{eq: est infty} and the area estimate $ |\cellj| = \OO(\eps^2|\log\eps|^2)$. 

The main result concerning the energy lower bound is the following

	\begin{pro}[Energy lower bound]
		\label{pro: lower bound}
		\mbox{}	\\
		If $ 1 < b < \theo^{-1} $ as $ \eps \to 0 $ then
		\beq
			\label{eq: lower bound}
			\gle \geq \frac{|\partial \Om| \eones}{\eps} + \OO(|\log\eps|^2).
		\eeq
	\end{pro}
	
The core of the proof is the same argument used in the proof of \cite[Proposition 4.2]{CR1}, but in order to get to the spot where one can apply the estimate of the cost function, few adjustments are in order. First of all the functional $ \glff $ is given on the right domain $ \acut $, where we can pass to tubular coordinates and replace the vector potential $ \aavm $, but because $ \acut $ is made of several connected components, we need to suitably modify $ \glm $ and impose its vanishing at the normal and inner boundaries of those sets. The reason of this will become clear only at a later stage of the proof: thanks to so-imposed Dirichlet boundary conditions, several unwanted boundary terms will vanish when integrating by parts the current term in the functional.

We sum up this preliminary steps in the following

	\begin{lem}
		\label{lem: lb restriction}
		\mbox{}	\\
		As $ \eps \to 0 $ 
		\beq
			\label{eq: lb restriction}
			\glff[\glm,\aavm] \geq \F[\psi] + \OO(|\log\eps|^2),	
		\eeq
		where 
		\beq
			\label{eq: F en}
			\F[\psi] : = \int_{\acut} \diff s \diff t \: \lf\{ \lf| \partial_t \psi \ri|^2 + \lf| \lf(\partial_s - i t \ri) \psi \ri|^2 - \tx\frac{1}{2b}[2|\psi|^2-|\psi|^4] \ri\},
		\eeq
		and, denoting $ \annt : = \{ (s,t) \in \ann \: | \: t \leq c_0 |\log\eps| - \eps \} $,
		\beq
			\label{eq: psi}
			\psi(s,t) : = 
			\begin{cases}
				\glm(\rv(\eps s, \eps t)) \: \exp \lf\{- i \phi_{\aavm}(s,t) \ri\},	&	\mbox{in } \annt \setminus \cup_j \celldj,	\\
				0,	&	\mbox{for } s = s_j \pm c_1 |\log\eps|, 	\\
				0,	&	\mbox{for } t = c_0 |\log\eps|,
			\end{cases}
		\eeq
		and $ |\psi| \leq | \glm | $ everywhere.
	\end{lem}

	\begin{proof}
		We first pass to boundary coordinates and simultaneously replace the magnetic potential $ \aavm $ as described in Section \ref{sec: replacement}: this leads to the lower bound
		\bml{
			\label{eq: lb restriction eq 1}
			\glff[\glm,\aavm] \geq  \int_{\acut} \diff s \diff t \: \lf\{ \big| \partial_t \tilde\psi \big|^2 + \big| \lf(\partial_s + i a_{\aavm}(s,t) \ri) \tilde\psi \big|^2 - \tx\frac{1}{2b}\lf[2|\tilde\psi|^2-|\tilde\psi|^4 \ri] \ri\}	\\
			+ \OO(|\log\eps|),
		}
		where $ \tilde \psi (s,t) = \glm(\rv(\eps s, \eps t)) \: \exp \lf\{- i \phi_{\aavm}(s,t) \ri\} $ and $ a_{\aavm} $ is given in \eqref{eq: magn pot a}. The remainder $ \OO(|\log\eps|) $ is the product of the prefactor $ \OO(\eps|\log\eps|) $ due to the jacobian of the coordinate transformation times the negative term proportional to the $ L^2$ norm of $ \tilde \psi $.
		
		Next acting as in \cite[Eq. (4.26)]{CR1} and using Lemma \ref{lem: a est}, we can estimate for any $ \delta >  0 $,
		\bmln{
			 \int_{\acut} \diff s \diff t \: \lf[ \big| \lf(\partial_s + i a_{\aavm}(s,t) \ri) \tilde\psi \big|^2 - \big| \lf(\partial_s - i t) \ri) \tilde\psi \big|^2 \ri] \\
			 \geq - \delta \lf\| \lf( \nabla + i \tx\frac{\aavm}{\eps^2} \ri) \glm \ri\|_{L^2(\Omega)}^2 - \lf( \tx \frac{1}{\delta} + 1 \ri) \lf\|  a_{\aavm}(s,t) + t \ri\|_{L^2(\acut)}^2	\\
			 \geq  - C \lf[ \delta \eps^{-1} + \delta^{-1} \eps^2 |\log\eps|^2 \ri] \geq - C \sqrt{\eps} |\log\eps|,
		}
		after an optimization over $ \delta $. Hence we get from \eqref{eq: lb restriction eq 1}
		\beq
			\label{eq: lb restriction eq 2}
			\glff[\glm,\aavm] \geq \F\big[\tilde\psi\big] + \OO(|\log\eps|).
		\eeq
		
		To impose the boundary conditions at the normal and inner boundaries of $ \acut $, we use two different partition of unity, i.e., two pairs of smooth functions $ 0 \leq \chi_i, \eta_i \leq 1 $, $ i = 1, 2 $, such that $ \chi_i^2 + \eta_i^2 = 1 $ and
		\beq
			\chi_1 = \chi_1(s) =
			\begin{cases}
				1,	&	\mbox{in } \ann \setminus \cup_j \celldj,	\\
				0,	&	\mbox{in } \cup_j \cellj,			
			\end{cases}
		\eeq
		\beq
			\chi_2 = \chi_2(t) =
			\begin{cases}
				1,	&	\mbox{for } t \in [0, c_0 |\log\eps| - \eps],	\\
				0,	&	\mbox{for } t = c_0|\log\eps|.		
			\end{cases}
		\eeq
		Given the size where $ \chi_i,  \eta_i $ are not constant, we can assume the following estimates to hold true
		\beq
			\label{eq: grad chi 1}
			\lf| \nabla \chi_1 \ri| = \OO(|\log\eps|^{-1}),	\qquad		\lf| \nabla \eta_1 \ri| = \OO(|\log\eps|^{-1}),
		\eeq
		\beq
			\label{eq: grad chi 2}
			\lf| \nabla \chi_2 \ri| = \OO(\eps^{-1}),	\qquad		\lf| \nabla \eta_2 \ri| = \OO(\eps^{-1}).
		\eeq
		The IMS formula then yields
		\bml{
			\label{eq: lb restriction eq 3}
			\F \big[\tilde\psi\big] \geq \F \big[\chi_1 \chi_2 \tilde{\psi}\big] - \int_{\ann} \diff s \diff t \: \lf[ {\chi^{\prime}_1}^2+ {\eta_1^{\prime}}^2 \ri] \big| \tilde\psi \big|^2 \\
			- \int_{\ann} \diff s \diff t \: \lf[ {\chi^{\prime}_2}^2+ {\eta_2^{\prime}}^2 \ri] \big| \tilde\psi \big|^2 + \OO(|\log\eps|^2),
		}
		where we have estimated
		\bdm
			\F\big[\eta_1 \chi_2 \tilde\psi] + \F\big[\eta_2 \chi_1 \tilde\psi] + \F\big[\eta_1 \eta_2 \tilde\psi] \geq - \frac{1}{b} \int_{\ann} \diff s \diff t \: \lf[ \eta_1^2 + \eta^2_2 + \eta_1^2 \eta_2^2 \ri] \big|\tilde\psi\big|^2 \geq - C |\log\eps|^2.
		\edm
		Using \eqref{eq: grad chi 1} it is easy to show that the second term in \eqref{eq: lb restriction eq 3} can be absorbed in the remainder, while, thanks to Agmon estimates,
		\bdm
			\int_{c_0|\log\eps| - \eps}^{\bar{t}} \diff t \int_{0}^{\frac{|\partial \Omega|}{\eps}} \diff s \: \big| \tilde \psi \big|^2 \leq \int_{\dist(\rv, \partial \Omega) \geq c_0 \eps|\log\eps| - \eps^2} \diff \rv \: \lf| \glm \ri|^2 = \OO(\eps^{c_0 C_A+1}),
		\edm
		i.e., $ \tilde\psi $ is still smaller than any power of $ \eps $ in the support of $ \chi_2^{\prime} $ and $ \eta_2^{\prime} $, which implies that the third term in  \eqref{eq: lb restriction eq 3} can be discarded as well. 
		
		In conclusion we obtained
		\beq
			\F \big[\tilde\psi\big] \geq \F \big[\chi_1 \chi_2 \tilde{\psi}\big] + \OO(|\log\eps|^2),
		\eeq
		and, setting $ \psi : = \chi_1 \chi_2 \tilde{\psi} $, the claim is proven.
	\end{proof}

	The rest of the lower bound proof is very close to the proof of \cite[Proposition 4.2]{CR1}. We sum up the main steps below.

	\begin{proof}[Proof of Proposition \ref{pro: lower bound}]
		Combining \eqref{eq: restriction} with \eqref{eq: pre lower bound} and the result of Lemma \ref{lem: lb restriction}, we have
		\beq
			\label{eq: lower bound eq 1}
			\glee \geq \F[\psi] + \OO(|\log\eps|^2).
		\eeq
		The next step is thus a lower bound to $ \F[\psi] $.
		
		First of all we extract from $ \F[\psi] $ the desired leading term in the energy asymptotics: by a standard splitting trick, we set 
		\beq
			\psi(s,t) = : \fs(t) u(s,t) e^{-i \as s},
		\eeq
		which defines a suitable $ u \in H^1_{\mathrm{loc}}(\acut) $. Note that, since $ \as $ is in general not an integer, $u $ is not periodic and therefore a multi-valued function, but $ |u| $ is periodic and this will suffice. Plugging the above ansatz in the functional $ \F $, we get
		\beq
			\label{eq: splitting}
			\F[\psi] = \frac{| \partial \ocut | \eones}{\eps} + \E[u],
		\eeq
		where
		\beq
			\E[u]:= \displaystyle\int_{\acut} \diff s \diff t\; \fs^2 \lf\{ |\partial_t u|^2 + |\partial_s u|^2 - 2 (t+ \as) {\bf e}_s \cdot \jv[u] + \tx\frac{1}{2b} \fs^2 \lf(1-|u|^2 \ri)^2 \ri\},
		\eeq
		and the superconducting current is given by
		\beq
			\jv[u] := \tx{\frac{i}{2}} \lf(u\nabla u^* - u^*\nabla u\ri).
		\eeq
		Since 
		\bdm
			|\partial \ocut| = |\partial \Omega| + \OO(\eps|\log\eps|),
		\edm
		the lower bound is proven if we can show that $ \E[u] \geq 0 $. The rest of the proof is focused on this claim.
		
		In order to investigate the positivity of $ \E[u] $ we use the potential function trick, i.e., we observe that the function $ F $ defined in \eqref{eq: F} satisfies
		\beq
			F^{\prime}(t)= 2 (t+ \as) \fs^2(t),
		\eeq
		and therefore
		\bdm
			 -2 \int_{\acut} \diff s\diff t \; (t+ \as) j_s(u) = - \int_{\acut} \diff s\diff t \; \partial_t F (t)\; j_s(u) =  \int_{\acut} \diff s\diff t \; F(t) \; \partial_t j_s[u]
		\edm
		where we have denote by $ j_s[u] = {\bf e}_s \cdot \jv[u] $ the $s-$component of the current. Here the boundary terms vanish because $ F(0) = 0$ and
		\beq
			\label{eq: vanishing u}
			u(s,c_0|\log\eps|) = 0,	\qquad		u(s_j \pm c_1 |\log\eps|, t) = 0,
		\eeq
		thanks to the boundary conditions inherited from $ \psi $ and the strict positivity of $ \fs $.

		We now integrate by parts in the $s$ variable the last two terms:
		\bmln{
			 \int_{\acut} \diff s\diff t \; F(t) \; \partial_t j_s[u] =  \tx\frac{i}{2} \disp\int_{\acut} \diff s\diff t \; F(t) \; \lf[ \partial_t u \partial_s u^* - \partial_t u^* \partial_s u + u \partial^2_{s,t} u^* - u^* \partial^2_{s,t} u \ri]  \\
			 = i \int_{\acut} \diff s\diff t \; F(t) \; \lf[ \partial_t u \partial_s u^* - \partial_t u^* \partial_s u  \ri],
		}
		where again boundary terms are absent thanks to the vanishing of $ u $ stated in \eqref{eq: vanishing u}. At this stage the non-periodicity of $ u $ could affect the result but this is not the case because $ u^* \partial_t u $ and its complex conjugate are always periodic. The simple estimate
		\bmln{
			 i \int_{\acut} \diff s\diff t \; F(t) \; \lf[ \partial_t u \partial_s u^* - \partial_t u^* \partial_s u  \ri] \geq - 2\displaystyle\int_{\acut} \diff s \diff t\; |F(t)||\partial_t u||\partial_s u|	\\
		 \geq\displaystyle\int_{\acut}\diff s\diff t\; F(t) \;[|\partial_t u|^2 +|\partial_s u|^2]
		}
		which uses the negativity of $ F $, then leads us to the lower bound for $ \E[u] $:
		\beq
			\E[u] \geq \displaystyle\int_{\acut} \diff s\diff t\; \lf\{ K(t) \lf( |\partial_t u|^2+|\partial_s u|^2 \ri) + \frac{1}{2b} \fs^4(1-|u|^2)^2 \ri\}.
		\eeq
		The pointwise positivity of $ K(t) $ for $ 1 < b < \theo^{-1} $ given in \eqref{eq: K positive} and the manifest positivity of the second term in the expression above yields the final lower bound
		\beq
			\label{eq: last lb}
			\E[u] \geq \frac{1}{2b} \int_{\acut} \diff s \diff t \: \fs^4(1-|u|^2)^2 \geq 0.
		\eeq
	\end{proof}
	
	We finalize now the proof of the main result:
	
	\begin{proof}[Proof of Theorem \ref{teo: GL asympt}]
		The combination of the energy upper (Proposition \ref{pro: upper bound}) and lower (Proposition \ref{pro: lower bound}) bounds yields the energy asymptotics \eqref{eq: gle asympt}. It only remains to prove the estimate on the $ L^2 $ norm of the difference $ |\glm|^2 - \fs^2 $. This is however trivially implied by the lower bound \eqref{eq: last lb}: if one keeps the positive term appearing on the r.h.s. of the inequality and put it together with \eqref{eq: lower bound eq 1}, the splitting \eqref{eq: splitting} and the upper bound \eqref{eq: upper bound}, the outcome is
		\beq
			\int_{\acut} \diff s \diff t \: \fs^4(1-|u|^2)^2 = \OO(|\log\eps|^2).
		\eeq
		By reconstructing first $ \psi(s,t)$ and then using \eqref{eq: psi}, one can easily realize that the regions where $ |\psi| $ differs from $ |\glm| $ can be discarded and their contribution be included in the remainder. The same holds true for the corner cells and therefore the final result is \eqref{eq: glm asympt}. Note the factor $ \eps^2 $ appearing on the r.h.s. due to the rescaling $ (s,t) \to (\sigma,\tau) = (\eps s, \eps t) $.
	\end{proof}

\end{document}